\documentclass[preprint]{elsarticle}

\usepackage{graphicx,url}

\usepackage[utf8]{inputenc}  
\usepackage{amsmath}
\usepackage{amssymb}
\usepackage{amsthm}
\usepackage{latexsym}
\usepackage{float}

\allowdisplaybreaks

\journal{Science of Computer Programming}

\begin{document}

\newcommand{\Nothing}{{\tt fail}}
\newcommand{\Just}[1]{#1}
\newcommand{\PE}[1]{\ensuremath{\mathcal{T}(#1)}}
\newcommand{\RE}[1]{\ensuremath{\mathcal{R}(#1)}}
\newcommand{\fivespaces}{\;\;\;\;\;}
\newcommand{\tenspaces}{\fivespaces\fivespaces}
\newcommand{\fifteenspaces}{\tenspaces\fivespaces}
\newcommand{\twentyspaces}{\tenspaces\tenspaces}
\newcommand{\thirtyspaces}{\twentyspaces\tenspaces}
\newcommand{\fortyspaces}{\twentyspaces\twentyspaces}
\newcommand{\fiftyspaces}{\fortyspaces\tenspaces}
\newcommand{\sixtyspaces}{\fortyspaces\twentyspaces}
\newcommand{\interf}{\fivespaces}
\newcommand{\mylabel}[1]{\, \mathbf{(#1)}}
\newcommand{\Con}[2]{#1\,#2}
\newcommand{\Choice}[2]{#1\,|\:#2}
\newcommand{\Chre}[2]{#1\,|\,#2}
\newcommand{\Choexe}[2]{#1\;\;|\;\;#2}
\newcommand{\Ochoexe}[2]{#1\;\;/\;\;#2}
\newcommand{\Ocho}[2]{#1\,/\:#2}
\newcommand{\Ore}[2]{#1\,/\,#2}
\newcommand{\Gp}{G^\prime}
\newcommand{\Gk}{G_k}
\newcommand{\Gpk}{G^\prime_k}
\newcommand{\Xp}{x^\prime}
\newcommand{\Xpp}{x^{\prime\prime}}
\newcommand{\Yp}{y^\prime}
\newcommand{\Ypp}{y^{\prime\prime}}
\newcommand{\Zp}{z^\prime}
\newcommand{\Wp}{w^\prime}
\newcommand{\Anyx}{X}
\newcommand{\Anyo}{o}
\newcommand{\Bvert}{\big\vert}
\newcommand{\Lp}{\stackrel{\mbox{\tiny{PEG}}}{\leadsto}}
\newcommand{\Lg}{\ensuremath{\stackrel{\mbox{\tiny{CFG}}}{\leadsto}}}
\newcommand{\Ll}{\stackrel{\mbox{\tiny{L\!L\!(1)}}}{\leadsto}}
\newcommand{\Lk}{\stackrel{\mbox{\tiny{L\!L\!(k)}}}{\leadsto}}
\newcommand{\Lr}{\stackrel{\mbox{\tiny{R\!E}}}{\leadsto}}
\newcommand{\Epsi}{\varepsilon}
\newcommand{\Rp}[2]{\Pi(#1,\,#2)}
\newcommand{\Tup}[2]{(#1,\,#2)}
\newcommand{\Peg}[2]{#1[#2]}
\newcommand{\Prod}[2]{\Peg{#1}{#2}}
\newcommand{\Pgg}[1]{\Peg{G}{#1}}
\newcommand{\Pgk}[1]{\Peg{\Gk}{#1}}
\newcommand{\Grm}[4]{(#1,\,#2,\,#3,\,#4)}
\newcommand{\Mat}[2]{#1\;\,#2\,}
\newcommand{\Reg}[2]{#1\;\,#2\,}
\newcommand{\Matg}[2]{\Mat{\Pgg{#1}}{#2}}
\newcommand{\Matk}[2]{\Mat{\Pgk{#1}}{#2}}
\newcommand{\Regk}[1]{\Reg{\Ek}{#1}}
\newcommand{\Pk}{p_k}
\newcommand{\Vk}{V_k}
\newcommand{\Ek}{e_k}
\newcommand{\Len}[1]{\Bvert#1\Bvert}
\newcommand{\Ps}{p_S}
\newcommand{\Tet}[1]{\texttt{#1}}
\newcommand{\Ch}[1]{\texttt{#1}}
\newcommand{\Chh}[1]{\mbox{`}{#1}\mbox{'}}
\newcommand{\FLW}{FOLLOW^G}
\newcommand{\FST}{FIRST^G}
\newcommand{\arrow}{\rightarrow}
\newcommand{\Ep}{e^{\prime}}
\newcommand{\Rig}{\Rightarrow_G}
\newcommand{\Mrig}[4]{\Tup{#1}{#2} \Rig \Tup{#3}{#4}}
\newcommand{\Prig}[3]{\Tup{#1}{#2} \Rig^+\, #3}
\newcommand{\Rigs}{{\stackrel{*}{\Rightarrow}}_G}
\newcommand{\Gtdpl}[3]{#1[#2,\,#3]}
\newcommand{\Lcfg}{L^{\scriptscriptstyle C\!F\!G}(G)}
\newcommand{\Lpeg}{L^{\scriptscriptstyle P\!E\!G}(G)}
\newcommand{\Lpegp}{L^{\scriptscriptstyle P\!E\!G}(\Gp)}
\newcommand{\Lford}{L^{\scriptscriptstyle F}(G)}
\newcommand{\Pdot}{\textbf{.}}
\newcommand{\Dq}{\texttt{"}}
\newcommand{\Sq}{\texttt{'}}

\newcommand{\Mend}{matchEnd}
\newcommand{\Tend}{$\Epsi$-End}
\newcommand{\Pop}{p_1^\prime}
\newcommand{\Ptp}{p_2^\prime}
\newcommand{\Ctwo}{choice_{LL(1)}.2}
\newcommand{\Vark}{var_{LL(k)}.1}
\newcommand{\Varl}{var_{LL(1)}.1}

\newcommand{\Matxas}[2]{#1 \stackrel{\preceq x}{\rightarrow} #2}
\newcommand{\Matxpas}[2]{#1 \stackrel{\prec x}{\rightarrow} #2}
\newcommand{\Sucxw}[2]{#1 \stackrel{\preceq x}{\leftarrow} #2}
\newcommand{\Sucxpw}[2]{#1 \stackrel{\prec x}{\leftarrow} #2}
\newcommand{\Fout}{f_{out}}
\newcommand{\Fin}{f_{in}}
\newcommand{\Isnull}{isNull}
\newcommand{\Nnull}{\neg\Isnull}
\newcommand{\Hase}{hasEmpty}
\newcommand{\Nemp}{\neg\Hase}
\newcommand{\True}{\textbf{true}}
\newcommand{\False}{\textbf{false}}
\newcommand{\Lor}[2]{#1 \;\vee\; #2}
\newcommand{\Land}[2]{#1 \;\wedge\; #2}
\newcommand{\E}{e_0}

\newcommand{\FLWk}{FOLLOW^G_k}
\newcommand{\FLWt}{FOLLOW^G_2}
\newcommand{\FSTk}{FIRST^G_k}
\newcommand{\FSTt}{FIRST^G_2}
\newcommand{\Uk}{\otimes_k}
\newcommand{\Pp}{P^\prime}
\newcommand{\Matp}[2]{\Mat{\Pgp{#1}}{#2}}
\newcommand{\Pgp}[1]{\Peg{\Gp}{#1}}
\newcommand{\Strpeg}{string2peg}
\newcommand{\Setcho}{set2choice}
\newcommand{\fstk}{take_k}
\newcommand{\Pik}[2]{\Pi(#1,\,#2)}
\newcommand{\FLWs}{\Con{\$}{\$}}

\newcommand{\Myeq}{\;=\;}

\newtheorem{definition}{Definition}[section]
\newtheorem{proposition}{Proposition}[section]
\newtheorem{corollary}[proposition]{Corollary}
\newtheorem{lemma}[proposition]{Lemma}

\begin{frontmatter}

\title{On the Relation between Context-Free Grammars
and Parsing Expression Grammars}

\author{Fabio Mascarenhas}
\ead{fabiom@dcc.ufrj.br}
\address{Department of Computer Science -- UFRJ -- Rio de Janeiro -- Brazil} 

\author{S\'{e}rgio Medeiros}
\ead{sergiomedeiros@ect.ufrn.br}
\address{School of Science and Technology -- UFRN -- Natal -- Brazil} 

\author{Roberto Ierusalimschy}
\ead{roberto@inf.puc-rio.br}
\address{Department of Computer Science -- PUC-Rio -- Rio de Janeiro -- Brazil} 

\begin{abstract}
Context-Free Grammars (CFGs) and Parsing Expression Grammars (PEGs) have several similarities and a few differences in both their syntax and semantics, but they are usually presented through formalisms that hinder a proper comparison. In this paper we present a new formalism for CFGs that highlights the similarities and differences between them. The new formalism borrows from PEGs the use of {\em parsing expressions} and the recognition-based semantics. We show how one way of removing non-determinism from this formalism yields a formalism with the semantics of PEGs. We also prove, based on these new formalisms, how LL(1) grammars define the same language whether interpreted as CFGs or as PEGs, and also show how strong-LL($k$), right-linear, and LL-regular grammars have simple language-preserving translations from CFGs to PEGs. Once these classes of CFGs can be automatically translated to equivalent PEGs, we can reuse classic top-down grammars in PEG-based tools. 
\end{abstract}


\end{frontmatter}

\section{Introduction}

Context-Free Grammars (CFGs) are the formalism of choice for describing the syntax of programming languages. A CFG describes a language as the set of strings generated from the grammar's initial symbol by a sequence of rewriting steps. CFGs do not, however, specify a method for efficiently recognizing whether an arbitrary string belongs to its language. In other words, a CFG does not specify how to {\em parse} the language, an essential operation for working with the language (in a compiler, for example). Another problem with CFGs is ambiguity, where a string can have more than one parse tree.

Parsing Expression Grammars (PEGs)~\cite{pegford} are an alternative formalism for describing a language's syntax. Unlike CFGs, PEGs are unambiguous by construction, and their standard semantics is based on recognizing strings instead of generating them. A PEG can be considered both the specification of a language and the specification of a top-down parser for that language. 

The idea of using a formalism for specifying parsers is not new; PEGs are based on two formalisms first proposed in the early seventies, Top-Down Parsing Language (TDPL)~\cite{tdpl} and Generalized TDPL (GTDPL)~\cite{gtdpl}. PEGs have in common with TDPL and GTDPL the notion of {\em limited backtracking} top-down parsing: the parser, when faced with several alternatives, will try them in a deterministic order (left to right), discarding remaining alternatives after one of them succeeds. Compared with the older formalisms, PEGs introduce a more expressive syntax, based on the syntax of regexes, and add {\em syntactic predicates}~\cite{antlr:parr}, a form of unrestricted lookahead where the parser checks whether the rest of the input matches a parsing expression without consuming the input.

Ford~\cite{pegford} has already proven that PEGs can recognize any deterministic
context-free language, but leaves open the question of the relation between
context-free grammars and PEGs. In this paper, we argue that the similarities between CFGs and PEGs are deeper than usually thought, and how these similarities have been obscured by the way the two formalisms have been presented. PEGs, instead of a formalism completely unrelated to CFGs, can be seen as a natural outcome of removing the ambiguity of CFGs. 

We start with a new semantics for CFGs, using the framework of natural semantics~\cite{kahn,winskel}. The new semantics borrows the syntax of PEGs and is also based on recognizing strings. We make the source of the ambiguity of CFGs, their non-deterministic alternatives for each non-terminal, explicit in the semantics of our new non-deterministic choice operator. We then remove the non-determinism, and consequently the ambiguity, by adding explicit failure and ordered choice to the semantics, so now we can only use the second alternative in a choice if the first one fails. By that point, we only need to add the {\em not} syntactic predicate to arrive at an alternative semantics for PEGs (modulo syntactic sugar such as the repetition operator and the {\em and} syntactic predicate). We prove that our new semantics for both CFGs and PEGs are equivalent to the usual ones.

Our semantics for CFGs and PEGs make it clear that the defining characteristic that sets
PEGs apart from CFGs is the ordered choice. For example, the grammar $S \arrow \Con{(\Choice{aba}{a})}{b}$ is both a CFG and a PEG in our notation, but consumes the prefix
$ab$ out of the subject $abac$ when interpreted as a CFG, and fails as a PEG.

We also show in this paper how our new semantics for CFGs gives us a way to translate some subsets of CFGs to PEGs that parse the same language. The idea is that, as the sole distinction between CFG and PEG semantics is in the choice operator, we will have a PEG that is equivalent to the CFG whenever we can make the PEG choose the correct alternative at each choice, either through reordering or with the help of syntactic predicates. We show transformations from CFGs to PEGs for three unambiguous subsets of CFGs: LL(1), Strong LL($k$), and LL-regular.

A straightforward correspondence between LL(1) grammars and PEGs was already noted~\cite{compcompilers}, but never formally proven. The correspondence is that an LL(1) grammar describes the same language whether interpreted as a CFG or as a PEG. The intuition is that, if an LL(1) parser is able to choose an alternative with a single symbol of lookahead, then a PEG parser will fail for every alternative that is not the correct one. We prove that this intuition is correct if none of the alternatives in the CFG can generate the empty string, and also prove that a simple ordering of the alternatives suffices to hold the correspondence even if there are alternatives that can generate the empty string. In other words, any LL(1) grammar is already a PEG that parses the same language, modulo a reordering of the alternatives.

There is no such correspondence between strong-LL($k$) grammars and PEGs, not even by imposing a specific order among the alternatives for each non-terminal. Nevertheless, we also prove that we can transform a strong-LL($k$) grammar to a PEG, just by adding a predicate to each alternative of a non-terminal. We can either add a predicate to the beginning of each alternative, thus encoding the choice
made by a strong-LL($k$) parser in the grammar, or, more interestingly, add
a predicate to the end of each alternative.

Our transformations lead to efficient parsers for LL(1) and strong-LL($k$) grammars, even in PEG implementations that do not use memoization to guarantee O($n$) performance, because the resulting PEGs only use backtracking to test the lookahead of each alternative, so their use of backtracking is equivalent to a top-down parser checking the next $k$ symbols of lookahead against the lookahead values of each production.

There is no direct correspondence between LL-regular grammars and PEGs,
either, given that strong-LL($k$) grammars are a proper subset of LL-regular
grammars. But we also show that we can transform any LL-regular grammar
into a PEG that recognizes the same language: we first prove that right-linear grammars for languages with the prefix property, a property that is easy to achieve, have the same language whether interpreted as CFGs or as PEGs, then use this result to build lookahead expressions for the alternatives of each non-terminal based on which regular partition this alternative falls.

While LL(1) grammars are a proper subset of strong-LL($k$) grammars,
which are a proper subsets of LL-regular grammars, thus making the LL-regular
transformation work on grammars belonging to these simpler classes, the simpler
classes have more straightforward transformations which merit a separate
treatment.

Given that these classes of top-down CFGs can be automatically translated into
equivalent PEGs, we can reuse classic top-down grammars in PEG-based tools.
So grammars written for tools such as ANTLR~\cite{antlr:parr} could be reused in
a parser tool that has PEGs as its backend language. As PEGs are composable,
these grammars can then be used as components in larger grammars. The
language designer can then start with a simple, LL(1) or strong-LL($k$) subset
of the language, and then grow it into the full language.
 
The rest of this paper is organized as follows: Section 2 presents our new semantics for CFGs and PEGs, showing how to arrive at the latter from the former, and proves their correctness. Section 3 shows how an LL(1) grammar describes the same language when interpreted as a PEG, and proves this correspondence. Section 4 shows how a simple transformation generates a PEG from any strong-LL($k$) grammar, keeping the same general structure and describing the same language as the original grammar, and proves the latter assertion. Section 5 shows the equivalence between some right-linear CFGs and PEGs, and how this can be used to build a simple transformation that generates a PEG from an LL-regular grammar. Finally, Section 6 reviews related work, and Section 7 summarizes the paper's contributions and gives our final remarks.

\section{From CFGs to PEGs}
\label{sec:cfgpeg}

\newcommand{\matchg}{\ensuremath{\stackrel{\mbox{\tiny{CFG}}}{\leadsto}}}

This section presents a new definition of CFGs, based on natural semantics,
and shows how from it we can establish a correspondence between CFGs and PEGs,
which formalization is also given through natural semantics. 

The section is structured as follows: Subsection~\ref{subsec:pecfg} briefly reviews
the traditional definition of CFGs and presents our new definition of CFGs, that
is called PE-CFGs and uses parsing expressions, borrowed from PEGs, and natural semantics.
Subsection~\ref{subsec:cfgpecfg} shows how we can obtatin a PE-CFG from a CFG and proves
that both definitions are equivalent. Next, Subsection~\ref{subsec:pecfgpeg} discusses the
relationship between PE-CFGs and PEGs, and defines PEGs by adapting the formalization of PE-CFGs.
Finally, Subsection~\ref{subsec:pegford} proves that our definition of PEGs is equivalent to Ford's
definition.

\subsection{From CFGs to PE-CFGs}
\label{subsec:pecfg}

The traditional definition of a CFG is as a tuple $(V, T, P, S)$ of a finite set $V$ of non-terminals symbols, a finite set $T$ of terminal symbols, a finite relation $P$ between non-terminals and strings of terminals and non-terminals, and an initial non-terminal $S$. We say that $A \rightarrow \beta$ is a {\em production} of $G$ if and only if $(A, \beta) \in P$. 

A grammar $G$ defines a relation $\Rightarrow_{G}$ where $\alpha A \gamma \Rightarrow_{G} \alpha \beta \gamma$ if and only if $A \rightarrow \beta$ is a production of $G$. The language of G is the set of all strings of terminal symbols that relate to $S$ by the reflexive-transitive closure of $\Rightarrow_{G}$. We can interpret the relation $\Rightarrow_{G}$ as a rewriting step, and then the language of $G$ is the set of all strings of terminals that can be generated from $S$ by a finite number of rewriting steps.

We want to give a new definition for CFGs that is closer to PEGs, so the similarities between the two formalisms will be more visible. Our new definition begins by borrowing the concept of a {\em parsing expression} from PEGs. The abstract syntax of parsing expressions is given below:
\[
p \; = \;\,\Epsi \;\; \Bvert
       \;\;  a \;\;         \Bvert
       \;\;  A \;\;         \Bvert
       \;\; \Con{p_1}{p_2} \;\; \Bvert
		   \;\; \Choice{p_1}{p_2} 
\]

Parsing expressions are defined inductively as the empty expression $\varepsilon$, a terminal symbol $a$, a non-terminal symbol $A$, a concatenation $p_1p_2$  of two parsing expressions $p_1$ and $p_2$, or a {\em choice} $\Choice{p_1}{p_2}$ between two parsing expressions $p_1$ and $p_2$.

We now define a PE-CFG (short for {\em CFG using parsing expressions}) $G$ as a tuple $(V, T, P, p_{S})$, where $V$ and $T$ are still the sets of non-terminals and terminals, but $P$ is now a function from non-terminals to parsing expressions, and $p_{S}$ is the initial parsing expression of the grammar. As $P$ is a function, we will use the standard notation for function application, $P(A)$, to refer to the parsing expression associated with a non-terminal $A$ in $G$.

Instead of the relation $\Rightarrow_{G}$, we define a new relation, $\matchg$, among a grammar $G$, a string of terminal symbols $v$, and another string of terminal symbols $w$. We will use the notation $G \, v \matchg w$ to say that $(G, v, w) \in \, \matchg$. The intuition for the $\matchg$ relation is that the first string is the {\em input}, and the second string is a suffix of the input that is left after $G$ matches a prefix of this input. We will usually say $G \, xy \matchg y$ to mean that $G$ matches a prefix $x$ of input string $xy$.


\begin{figure*}[t]
{
\small
\begin{align*}
& \textbf{Empty} \fivespaces
{\frac{}{\Matg{\Epsi}{x} \Lg \Just{x}}} \mylabel{empty.1}
\tenspaces
\textbf{Terminal} \fivespaces
{\frac{}{\Matg{a}{ax} \Lg \Just{x}}} \mylabel{char.1}  \\ \\ 
& \textbf{Non-terminal} \;\;\;
{\frac{\Matg{P(A)}{xy} \Lg \Just{y}}
	{\Matg{A}{xy} \Lg \Just{y}}}  \mylabel{var.1} \fivespaces 
\\ \\ & \textbf{Concatenation} \;\;\;
{\frac{\Matg{p_1}{xyz} \Lg \Just{yz} \interf \Matg{p_2}{yz} \Lg \Just{z}}
{\Matg{\Con{p_1}{p_2}}{xyz} \Lg \Just{z}}} \mylabel{con.1} \\ \\
& \textbf{Choice} \tenspaces
{\frac{\Matg{p_1}{xy} \Lg \Just{y}}
	{\Matg{\Choice{p_1}{p_2}}{xy} \Lg \Just{y}}} \mylabel{choice.1} \tenspaces
{\frac{\Matg{p_2}{xy} \Lg \Just{y}}
	{\Matg{\Choice{p_1}{p_2}}{xy} \Lg \Just{y}}} \mylabel{choice.2} 
\end{align*}
\caption{Natural semantics of \matchg}
\label{fig:matchcfg}
}
\end{figure*}

Figure~\ref{fig:matchcfg} shows our semantics for $\matchg$ using natural semantics, as a set of inference rules. $G \, xy \matchg y$ if and only if there is a finite proof tree for it, built using these rules. The notation $G[p^{\prime}_{S}]$ denotes a new grammar $(V, T, P, p^{\prime}_{S})$ that is equal to $G$ except for the initial parsing expression $p_{S}$, which is replaced by $p^{\prime}_{S}$. Each rule follows naturally from the intuition of $\matchg$: an empty parsing expression does not consume any input ({\bf empty.1}); a terminal consumes itself if it is the first symbol of the input ({\bf char.1}); a non-terminal matches its corresponding production in $P$ ({\bf var.1}); a concatenation first matches $p_1$ and then matches $p_2$ with what is left of the input ({\bf con.1}); and a choice can match either $p_1$ or $p_2$ ({\bf choice.1} and {\bf choice.2}, respectively). The rules guarantee that if $G \, v \matchg w$ then $w$ is a suffix of $v$.

The language of $G$, $L(G)$, is now the set of prefixes that $G$ matches, that is, all strings $x$ where $G \, xy \matchg y$ for some string $y$. In the traditional definition of CFGs, the language of a grammar is the set of strings the grammar generates; in our new definition, the language is the set of strings the grammar matches. We could have defined the language of $G$ as the set of strings $x$ where $G \, x \matchg \varepsilon$, that is, the set of strings that $G$ matches completely, but it is a corollary of the following lemma that the two definitions are equivalent:

\begin{lemma}
\label{cfg:lang}
Given a PE-CFG $G$, if $G[p] \, xy \matchg y$ then we have $\forall y^{\prime} . G[p] \, xy^{\prime} \matchg y^{\prime}$.
\end{lemma}

\begin{proof} By induction on the height of the proof tree for $G[p] \, xy \matchg y$.
\end{proof}

The previous lemma shows that the suffix in the relation $\matchg$ is superfluous; we could have defined relation $\matchg$ as a binary relation between a grammar $G$ and an input $w$, meaning just {\em $G$ recognizes $w$}. We chose to keep the suffix to emphasize the similarities between this semantics and our semantics for PEGs, where the suffix matters.

\subsection{Correspondence between CFGs and PE-CFGs}
\label{subsec:cfgpecfg}

We need a way to systematically transform a traditional CFG $G$ to a corresponding PE-CFG $G^\prime$, and vice-versa. To simplify our proofs, we will assume that grammars do not have useless symbols. The main obstacle for these transformations is the type of $P$. $P$ is a relation for CFGs, with different productions for the same non-terminal being different entries in this relation. In PE-CFGs, however, $P$ is a function, with all the different productions encoded as choices in the parsing expression for the non-terminal.

The choice operator is commutative, associative, and idempotent; both left and right concatenations distribute over choice, that is, $p_1(\Choice{p_2}{p_3}) = \Choice{p_1p_2}{p_1p_3}$ and $(\Choice{p_1}{p_2})p_3 = \Choice{p_1p_3}{p_2p_3}$\footnote{The proof of these properties is straightforward from the semantics of both operators.}. So any parsing expression may be rewritten as a choice $\Choice{p_1}{\Choice{\ldots}{p_n}}$, where the subexpressions $p_1, \ldots, p_n$ are distinct and do not have choice operators. We can then go from a PE-CFG $G^\prime$ to a CFG $G$ using $A \rightarrow p_1, \ldots, A \rightarrow p_n$ as the productions of each non-terminal $A$, where $p_1, \ldots, p_n$ are the subexpressions obtained by rewriting the expression $P^\prime(A)$ in the way above. 

Going from a CFG $G$ to a PE-CFG $G^\prime$ is easier: the right side of each production of $G$ is a concatenation of non-terminals and terminals, which translates directly to a concatenation of parsing expressions (the concatenation of expressions is associative); we assign an arbitrary order to the productions of each non-terminal $A$ of $G$, and then combine these productions right-associatively into a choice expression, and this is $P^\prime(A)$. We will call the transformation of CFGs to PE-CFGs $\mathcal{T}$, so $\mathcal{T}(G) = G^\prime$.

As an example, take the CFG $G$ with the following set of productions:
\[
P \Myeq \{\, A \arrow \Con{B}{C},\;
           B \arrow a,\; B \arrow b,\;
           C \arrow c,\; C \arrow d,\; C \arrow e \,\}
\]

Its corresponding PE-CFG $\mathcal{T}(G) = G^\prime$ has the following definition for the function $P^{\prime}$:
\[
P^{\prime}(A) \Myeq \Con{B}{C}    \tenspaces 
P^{\prime}(B) \Myeq \Choice{a}{b} \tenspaces
P^{\prime}(C) \Myeq \Choice{c}{\Choice{d}{e}}
\]

We used the order that we listed the productions of $G$ to order the choices, but commutativity and associativity of the choice operator guarantees that any other order would yield a grammar with the same language as $G^{\prime}$, so we could have used the following definition for $P^{\prime}$ instead:
\[
P^{\prime}(A) \Myeq \Con{B}{C}    \tenspaces 
P^{\prime}(B) \Myeq \Choice{a}{b} \tenspaces
P^{\prime}(C) \Myeq \Choice{e}{\Choice{c}{d}}
\]

The proof that $G$ and \PE{G} define the same language for any CFG $G$ is a direct corollary of the following lemma:

\begin{lemma}
\label{lemma:cfgequiv}
Given a CFG $G$ and its corresponding PE-CFG $\PE{G} = G^\prime$, we have $\alpha \stackrel{*}{\Rightarrow}_{G} x$ if and only if $G^{\prime}[\alpha] \, xy \matchg y$, where $x$ is a string of terminals and $\alpha$ is a string of terminals and non-terminals.
\end{lemma}

\begin{proof} ($\Rightarrow$) By induction on the number of steps in the derivation of $x$. The base case, where $\alpha = x$, is trivial, with an application of the {\bf empty.1} rule or repeated applications of the {\bf con.1} and {\bf char.1} rules. 

The induction step has $\alpha$ composed of three parts: a prefix $\alpha^{\prime}$, a non-terminal $A$, and a suffix $\gamma$, with $\alpha^{\prime} A \gamma \Rightarrow_{G} \alpha^{\prime} \beta \gamma \stackrel{*}{\Rightarrow}_{G} x$. By the properties of $\stackrel{*}{\Rightarrow}_{G}$, $x$ can be decomposed into $x_1$, $x_2$ and $x_3$ with $\alpha^{\prime} \stackrel{*}{\Rightarrow}_{G} x_1$, $\beta \stackrel{*}{\Rightarrow}_{G} x_2$, and $\gamma \stackrel{*}{\Rightarrow}_{G} x_3$. By the induction hypothesis we have $G^{\prime}[\alpha^{\prime}] \, x_1x_2x_3y \matchg x_2x_3y$,
$G^{\prime}[\beta] \, x_2x_3y \matchg x_3y$, and $G^{\prime}[\gamma] \, x_3y \matchg y$. We combine these proof trees in a proof tree for $G^{\prime}[\alpha^{\prime} A \gamma] \, xy \matchg y$ with rules {\bf con.1}, {\bf var.1}, and applications of the {\bf choice} rules to select the alternative corresponding to production $A \rightarrow \beta$.

($\Leftarrow$) By induction on the height of the proof tree for $G^{\prime}[\alpha] \, xy \matchg y$. The interesting case is {\bf var.1}; we need to use the fact that the use of choice operators in $G^{\prime}$ follows a known structure, where each production is a right-associative choice of parsing expressions that do not have choice operators and correspond to the right side of productions in $G$. So the proof tree for $G^\prime[P^{\prime}(A)] \, xy \matchg y$ ends with a succession of {\bf choice} rules that select which of the alternatives is taken for that non-terminal. We apply the induction hypothesis to the subtree above the last {\bf choice} rule used, from the consequent to the antecedents.
\end{proof}

A corollary of Lemma~\ref{lemma:cfgequiv} is that $S \stackrel{*}{\Rightarrow}_{G} x$ if and only if $\PE{G} \, xy \matchg y$, so the language of $G$ and the language of $\PE{G}$ are the same.

A traditional CFG is ambiguous if and only if there is some string with more than one leftmost (or rightmost) derivation. We can define ambiguity for PE-CFGs via proof trees: a PE-CFG $G$ is ambiguous if and only if there is more than one proof tree for $G \, xy \matchg y$ for some $x$ and $y$.

We can show that a CFG $G$ is ambiguous if and only if its corresponding PE-CFG $\PE{G}$ is ambiguous. The proof is a corollary of the proposition that there is only one leftmost derivation for $\alpha \stackrel{*}{\Rightarrow}_{G} x$ if and only if there is only one proof tree for $G^{\prime}[\alpha] \, xy \matchg y$, where $x$ is a string of terminals and $\alpha$ is a string of non-terminals and terminals. This proposition has a straightforward proof by induction (on the number of steps in the derivation and on the height of the proof tree), and the corollary follows by denial of the consequent.

\newcommand{\matchgd}{\ensuremath{\stackrel{\mbox{\tiny{C\!F\!G}}_d}{\leadsto}}}

Ambiguity, in our semantics, is directly tied to the choice operator: if we try to prove that there cannot be more than one proof tree for a $G[p] \, xy \matchg y$, by induction on the height of the tree, our proof fails for case {\bf choice.1}, because even if there is only one proof tree for the $G[p_1] \, xy \matchg y$, we might have $G[p_2] \, xy \matchg y$, so we can get another proof tree for $G[\Choice{p_1}{p_2}] \, xy \matchg y$ by using {\bf choice.2}. The proof fails for case {\bf choice.2} in a similar way.

If we can change the semantics of choice so that a single proof tree for its antecedents guarantees a single proof tree for the choice then we will guarantee that all grammars will be unambiguous. Obviously we will not have CFGs anymore; in particular, we will invalidate Lemma~\ref{lemma:cfgequiv}. In fact, our changes will take us from CFGs to a restricted form of PEGs, and we will prove that our changed semantics is equivalent to the semantics of PEGs as defined by Ford~\cite{pegford}.

\subsection{From PE-CFGs to PEGs}
\label{subsec:pecfgpeg}

\newcommand{\matchpeg}{\ensuremath{\stackrel{\mbox{\tiny{PEG}}}{\leadsto}}}

Now we will discuss in detail how we can obtain the semantics of PEGs
by changing the semantics of $\Lg$ presented in Figure~\ref{fig:matchcfg}.
 
In order to define the semantics of PEGs, we will make the choice operator {\em ordered}: in a choice $\Choice{p_1}{p_2}$ we try $p_2$ only if $p_1$ does not match. But we need a way to have a proof tree for ``$p_1$ does not match'', so we will also introduce an explicit failure result, $\Nothing$, to indicate the cases where a match is not possible. We will combine these changes in the semantics of a new relation $\matchpeg$. Figure~\ref{fig:matchcfgd} lists its inference rules, where $\Anyx$ means either $\Nothing$ or the remainder of the input string in a successful match.

\begin{figure*}[t]
{\small
\begin{align*}
& \textbf{Empty} \fivespaces
{\frac{}{\Matg{\Epsi}{x} \matchpeg \Just{x}}} \mylabel{empty.1} 
\tenspaces
\textbf{Non-terminal} \fivespaces
{\frac{\Matg{P(A)}{x} \matchpeg \Just{\Anyx}}
	{\Matg{A}{x} \matchpeg \Just{\Anyx}}}    \mylabel{var.1}  \\ \\ 
& \textbf{Terminal} \;\;\;
{\frac{}{\Matg{a}{ax} \matchpeg \Just{x}}} \mylabel{char.1} \;\;\;
{\frac{}{\Matg{b}{ax} \matchpeg \Nothing}} \mbox{ , } b \neq a \mylabel{char.2} \\
& \fivespaces \fivespaces \fivespaces \; \; \; \; {\frac{}{\Matg{a}{\Epsi} \matchpeg \Nothing}} \mylabel{char.3}  \\ \\
& \textbf{Concatenation}
\fivespaces
{\frac{\Matg{p_1}{xy} \matchpeg \Just{y} \interf \Matg{p_2}{y} \matchpeg \Just{\Anyx}}
	{\Matg{\Con{p_1}{p_2}}{xy} \matchpeg \Just{\Anyx}}} \mylabel{con.1} \\
& \fivespaces \fivespaces \fivespaces \fivespaces \fivespaces \fivespaces \; {\frac{\Matg{p_1}{x} \matchpeg \Nothing}
	{\Matg{\Con{p_1}{p_2}}{x} \matchpeg \Nothing}} \mylabel{con.2} \\ \\
& \textbf{Ordered Choice} 
\fivespaces
{\frac{\Matg{p_1}{xy} \matchpeg \Just{y}}
	{\Matg{\Choice{p_1}{p_2}}{xy} \matchpeg \Just{y}}} \mylabel{ord.1} \\
& \fivespaces \fivespaces \fivespaces \fivespaces \fivespaces \fivespaces \;\;\;\;
 {\frac{\Matg{p_1}{xy} \matchpeg \Nothing \interf \Matg{p_2}{xy} \matchpeg \Just{y}}
	{\Matg{\Choice{p_1}{p_2}}{xy} \matchpeg \Just{y}}} \mylabel{ord.2} \\
& \fivespaces \fivespaces \fivespaces \fivespaces \fivespaces \fivespaces \;\;\;\;
 {\frac{\Matg{p_1}{x} \matchpeg \Nothing \interf \Matg{p_2}{x} \matchpeg \Nothing}
	{\Matg{\Choice{p_1}{p_2}}{x} \matchpeg \Nothing}} \mylabel{ord.3}
\end{align*}
\caption{Natural semantics of $\matchpeg$}
\label{fig:matchcfgd}
}
\end{figure*}

Just introducing $\Nothing$ does not change the semantics enough to be incompatible with regular CFGs; if we take the semantics of Figure~\ref{fig:matchcfgd} and replace rule {\bf ord.2} with {\bf choice.2} then we have a conservative extension of our PE-CFG semantics that introduces $\Nothing$, so all our previous proofs remain valid. Ordered choice, represented by rule {\bf ord.2}, is what changes the semantics so it is not representing CFGs anymore. A simple example that shows this change is the grammar $G$ below:
\[
S \rightarrow AB \tenspaces \fivespaces
A \rightarrow \Choice{aba}{a} \tenspaces \fivespaces
B \rightarrow b
\]

We have $G \, abac \matchg ac$, but $G \, abac \not\matchpeg ac$, as the only proof tree under $\matchpeg$ for the input string $abac$ is for $G \, abac \matchpeg \Nothing$. 

We will use  $L^{\mbox{\tiny{P\!E\!G}}}(G)$ for the language of a PE-CFG $G$ interpreted with  $\matchpeg$; as with $\matchg$, this is the set of strings $x$ for which there is a string $y$ with $G \, xy \matchpeg y$.  Informally, this set is still the set of all the prefixes that $G$ matches, only using $\matchpeg$ instead of $\matchg$. But there is no equivalent of Lemma~\ref{cfg:lang} for the $\matchpeg$ relation; for example, the grammar above matches $ab$ but fails for $abac$.

Properties of the operators also change under $\matchpeg$: the choice operator is not commutative, and concatenation does not distribute over choice on the right anymore (although it still distributes on the left).

In Section~\ref{sec:ll1peg}, we will show a class of PE-CFGs where $L(G) = L^{\mbox{\tiny{PEG}}}(G)$. For now, an interesting result is the following lemma, which proves that $L^{\mbox{\tiny{PEG}}}(G)$ is a subset of $L(G)$ for any PE-CFG $G$:

\begin{lemma}
\label{cfgd:subset}
Given a PE-CFG $G$, if $G[p] \, xy \matchpeg y$ then we have $G[p] \, xy \matchg y$.
\end{lemma}

\begin{proof} By induction on the height of the proof tree for $G[p] \, xy \matchpeg y$. The only rule that does not have an identical rule in $\matchg$ is {\bf ord.2}, but it can trivially be replaced by {\bf choice.2}.
\end{proof}

The intuition of $G \, x \matchpeg \Nothing$ is that $G$ does not match any prefix of $x$ (including the empty string). This is a corollary of the following lemma, which formally says that the result of $G[p] \, x$ is unique under $\matchpeg$, for any $G$, $p$ and $x$:

\begin{lemma}
\label{peg:unamb}
Given a PE-CFG $G$, if $G[p] \, x \matchpeg X$ and $G[p] \, x \matchpeg X^\prime$ then we have $X = X^\prime$, and there is only one proof tree for $G[p] \, x \matchpeg X$.
\end{lemma}

\begin{proof} By induction on the height of the proof tree for $G[p] \, x \matchpeg X$. The interesting cases are {\bf ord.1} and {\bf ord.2}; for {\bf ord.1}, the induction hypothesis rules out the possibility of $G[p_1] \, x \matchpeg \Nothing$, so {\bf ord.2} cannot apply even if we have $G[p_2] \, x \matchpeg X$. For {\bf ord.2}, we must have $G[p_1] \, x \matchpeg \Nothing$ by the induction hypothesis; even if $X$ is $\Nothing$ we cannot use rule {\bf ord.1}.
\end{proof}

The PE-CFGs that we will be dealing with in the rest of the paper will have an important property that becomes possible to express by introducing failure: they will be {\em complete} grammars~\cite{pegford}. A complete PE-CFG is one where for any expression $p$ and any input $x$ either $G[p] \, x \matchpeg x^\prime$ or $G[p] \, x \matchpeg \Nothing$. Ford~\cite{pegford} proves that any grammar that does not have direct or indirect left recursion (a property which can be structurally checked) is complete.

A PE-CFG $G$ is left-recursive when there is a non-terminal $A$ of $G$ and an input $x$ where trying to derive a proof tree for $G[A] \, x$ can make $G[A] \, x$ appear again higher up in the tree. Because the semantics of $\matchpeg$ is deterministic this means that a left-recursive PE-CFG may not have any proof tree for $G[A] \, x$ under $\matchpeg$; in this case, an implementation of PEGs that tries to match $x$ with the expression $G[A]$ will not terminate.

Once we have failure and unambiguity, it is natural to introduce a way to turn a failure into a success. This is the {\em not} syntactic predicate ($!p$ for any parsing expression $p$), which is a conservative extension of our semantics described in Figure~\ref{fig:not}.

\begin{figure*}[t]
{\small
\begin{align*}
& \textbf{Not Predicate} \tenspaces \fivespaces
{\frac{\Matg{p}{x} \Lp \Nothing} 
	{\Matg{!p}{x} \Lp \Just{x}}} \mylabel{not.1} \tenspaces
{\frac{\Matg{p}{xy} \Lp \Just{y}}
	{\Matg{!p}{xy} \Lp \Nothing}} \mylabel{not.2} 
\end{align*}
\caption{Natural Semantics of the {\em not} predicate}
\label{fig:not}
}
\end{figure*}

PE-CFGs extended with the not-predicate and interpreted using $\matchpeg$ are equivalent, syntactically as well as semantically, to PEGs. Ford~\cite{pegford} also included the repetition operator in the abstract syntax of PEGs, but eliminating  repetition is a simple matter of replacing each repetition expression $p^*$ with a new non-terminal $A_p$ with the production $A_p \rightarrow \Choice{pA_p}{\varepsilon}$, which is just a step up from simple syntactic sugar.

\subsection{Correspondence with Ford's Defintion}
\label{subsec:pegford}

Ford~\cite{pegford} defines the semantics of PEGs using a relation $\Rightarrow_{G}$ that is similar to $\matchpeg$. Unlike the relation $\Rightarrow_{G}$ for traditional CFGs, Ford's $\Rightarrow_{G}$ is not a single step in the match, but the whole match. The notation $(p, x) \Rightarrow_{G} (n, X)$, for $(p, x, n, X) \in \, \Rightarrow_{G}$, means that either the parsing expression $p$ matches the prefix $x^\prime$ of input $x$, if $X$ is $x^\prime$, or the match fails, if $X$ is $\Nothing$. The number $n$ is a step counter, used in proofs by induction involving $\Rightarrow_{G}$.

Ford's definition of relation $\Rightarrow_{G}$ is similar to our definition of the $\matchpeg$ relation, using a similar set of cases. The following lemma states that both definitions are equivalent:

\begin{lemma} Given a PE-CFG $G$ and a parsing expression $p$, 
$(p, xy) \Rightarrow_{G} (n, x)$ if and only if $G[p] \, xy \matchpeg y$ and
$(p, xy) \Rightarrow_{G} (n, \Nothing)$ if and only if $G[p] \, xy \matchpeg \Nothing$.
\end{lemma}

\begin{proof}
The proof of the ($\Rightarrow$) direction is a straighforward induction
on the step count $n$, while the proof ($\Leftarrow$) is a straighforward
induction on the height of the proof tree for $G[p] \, xy \matchpeg X^\prime$. 
\end{proof}

Our definition for the {\em language} of a PEG is different from Ford's, though. Ford defines  $L^{\mbox{\tiny{PEG}}}$ as the set of strings for which a PEG recognizes some prefix of the string, while we use the set of strings that the PEG recognizes. In particular, the language of $\varepsilon$ is $T^*$ by Ford's definition and $\varepsilon$ with ours.

Ordered choice is what makes PEGs essentially different from CFGs, though, so one way to go from a CFG that can be parsed top-down to a PEG that parses the same language, without changing the structure of the grammar, is to make sure that the PEG always chooses the correct alternative at each choice. In Sections~\ref{sec:ll1peg},~\ref{sec:llkpeg}, and 5 we show how this intuition leads to translations from three classes of CFGs for top-down parsing, LL(1), Strong LL($k$), and LL-regular, to equivalent PEGs.

\section{LL(1) Grammars and PEGs}
\label{sec:ll1peg}

LL(1) grammars are the subset of CFGs where a top-down parser can decide which production to use for a non-terminal by examining just the next symbol of the input. An LL(1) parser can then parse the whole input by starting with the initial non-terminal of the grammar and then choosing which production to apply, making a choice again whenever it encounters a non-terminal, without needing to backtrack on its choices. A correspondence between them and PEGs has already been noted~\cite{compcompilers}, but not formally proven, so they are a nice starting point for applying our new semantics of CFGs and PEGs to the task of finding translations from subsets of CFGs to corresponding PEGs.

We will divide this task in two parts: first we will consider LL(1) grammars without $\varepsilon$ expressions and show that there is a correspondence between these grammars and PEGs. Then we will consider grammars with $\varepsilon$ expressions and show that there is a correspondence between these grammars and PEGs if the ordering of the choice expressions respects a simple property.

In Section~\ref{sec:cfgpeg} we presented a method of translating a traditional CFG to a PE-CFG, a CFG using parsing expressions. That method generates PE-CFGs with a property that will be useful in the proofs for this section; because we are going to use this property in our proofs, we will formalize it with the following definition:

\begin{description}
\item[BNF structure] A PE-CFG $G = (V, T, P, p_s)$ has {\em BNF structure} if it obeys the following properties:
\begin{enumerate}
\item No choice expression of $G$ is part of a concatenation expression;
\item $p_s$ is a single non-terminal;
\item For every choice $\Choice{p_1}{p_2}$ of $G$, if $p_1$ matches the empty string then $p_2$ must also match the empty string.
\end{enumerate}
\end{description}

Any traditional CFG $G$ has a corresponding PE-CFG $G^\prime$ that has BNF structure; in particular, it is trivial to ensure that \PE{G}  always has BNF structure. Properties 1 and 2 of BNF structure are an obvious outcome of the transformation $\mathcal{T}$: the expression associated to each non-terminal is of the form $\Choice{p_1}{\Choice{\ldots}{p_n}}$, where the choices associate to the right and $p_1, \ldots, p_n$ do not have choice expressions, so property 1 applies; the initial expression of $G^\prime$ is the initial non-terminal of $G$, so property 2 also applies; finally, property 3 can is guaranteed by choosing an order for the productions of each non-terminal of $G$ so the productions that can generate the empty string are last.

Any PE-CFG without BNF structure also can be rewritten to have it, by distributivity of concatenation over choice on the left and on the right, associativity and commutativity of choice, and the addition of an extra non-terminal to be the start expression. Throughout the rest of this section we will only consider PE-CFGs that have BNF structure in our definitions and proofs.

A traditional CFG without $\varepsilon$ productions is LL(1) if and only if, for each of its non-terminals $A_i$, the {\em FIRST} sets for the right sides of the productions of $A_i$ are disjoint. Before we can give a definition for LL(1) PE-CFGs, we need to define what is the {\em FIRST} set of a parsing expression. 
We will use the following definition for the {\em FIRST} set of an expression $p$ with a PE-CFG $G$:
\[
\mathit{FIRST}^{G}(p) = \{ a \in T \, | \, G[p] \, axy \matchg y \}
\]

This definition of {\em FIRST} is equivalent to the definition for traditional CFGs for any parsing expression that does not have choice operators (that is, any parsing expression that has a corresponding string of terminals and non-terminals). We can use the PE-CFG to CFG equivalence lemma (Lemma~\ref{lemma:cfgequiv}) to conclude that $p \stackrel{*}{\Rightarrow}_{G} ax$ from $\PE{G}[p] \, axy \matchg y$, where $G$ is a traditional CFG. The {\em FIRST} set of $p$ in the traditional definition is the set $\{ a \in T \, | \, p \stackrel{*}{\Rightarrow}_{G} a\beta \}$. As we assumed in Section~2 that $G$ does not have useless symbols, this is the same as the set $\{ a \in T \, | \, p \stackrel{*}{\Rightarrow}_{G} ax \}$, and the two definitions of {\em FIRST} are equivalent.

We can now give a definition for LL(1) PE-CFGs without $\varepsilon$ expressions: a PE-CFG $G$ without $\varepsilon$ expressions is LL(1) if and only if, for every choice $\Choice{p_1}{p_2}$ in the grammar, the {\em FIRST} sets of $p_1$ and $p_2$ are disjoint.

It is straightforward to prove that a traditional CFG $G$ without $\varepsilon$ productions is LL(1) if and only if its corresponding PE-CFG \PE{G} is also LL(1). The proof uses property 1 of BNF structure, associativity of choice, and the property that $\mathit{FIRST}^{\PE{G}}(\Choice{p_1}{p_2}) = \mathit{FIRST}^{\PE{G}}(p_1) \cup \mathit{FIRST}^{\PE{G}}(p_2)$.

Now that we have a definition for LL(1) PE-CFGs without $\varepsilon$ expressions, we can show that these grammars can be interpreted as PEGs (by the relation $\matchpeg$) without changing their language. The assertion that the language of an LL(1) grammar is the same whether interpreted as a CFG or as a PEG is a corollary of the following lemma:

\begin{lemma} Given an LL(1) PE-CFG $G$ without $\varepsilon$ expressions, $G[p] \, xy \matchg y$ if and only if $G[p] \, xy \matchpeg y$.
\end{lemma}

\begin{proof}($\Rightarrow$) By induction on the height of the proof tree for $G[p] \, xy \matchg y$. The interesting case is {\bf choice.2}. For this case, we have $G[p_2] \, xy \matchg y$. As $G$ does not have $\varepsilon$ expressions, $x$ cannot be empty. Let $a$ be the first symbol of $x$. It is obvious that $a \in \mathit{FIRST}^G(p_2)$, so $a \notin \mathit{FIRST}^G(p_1)$ by the LL(1) property. So $G[p_1] \, xy \not\matchg w$, and, by denial of the consequent of Lemma~\ref{cfgd:subset}, $G[p_1] \, xy \not\matchpeg w$. LL(1) grammars cannot have left recursion~\cite{parsing}, so they are complete and  $G[p_1] \, xy \not\matchpeg w$ implies $G[p_1] \, xy \matchpeg \Nothing$. With the induction hypothesis and the application of {\bf ord.2} we have $G[\Choice{p_1}{p_2}] \, xy \matchpeg y$.

($\Leftarrow$) Just a special case of Lemma~\ref{cfgd:subset}.
\end{proof}

We will now show that a correspondence between LL(1) grammars and their corresponding PEGs still exists when we allow $\varepsilon$ expressions, as long as the LL(1) grammars have BNF structure. Grammars with $\varepsilon$ expressions can have $\varepsilon$ in the {\em FIRST} sets of their expressions, so we need a slightly different definition of {\em FIRST}:
\[
\mathit{FIRST}^{G}(p) = \{ a \in T \, | \, G[p] \, axy \matchg y \} \cup \mathit{nullable(p)}
\]
\[
\mathit{nullable(p)} = \left\{ \begin{array}{ll}
  \{ \varepsilon \} & \mbox{if $G[p] \, x \matchg x$} \\
  \emptyset & \mbox{otherwise}
  \end{array} \right.
\]

The LL(1) property for grammars with $\varepsilon$ expressions also uses a {\em FOLLOW} set, defined below:
\[
\begin{array}{l}
\mathit{FOLLOW}^{G}(A) = \{ a \in T \cup \{ \$ \} \, | \, \mbox{$G[A] \, yaz \matchg az$ is in a}\\
\tenspaces \tenspaces \tenspaces \tenspaces \tenspaces \,\,\,\mbox{proof tree for $G \, w\$ \matchg \$ $} \}
\end{array}
\]

Like with the {\em FIRST} set, it is straightforward to prove that our definition of {\em FOLLOW} is equivalent to the definition for traditional CFGs. The restriction involving the proof tree for $G \, w\$ \matchg \$ $ of our definition proceeds directly from the fact that the traditional definition only uses derivations starting from the initial symbol of the grammar, and CFG derivations correspond to PE-CFGs proof trees.

The general statement $G \, xy \matchg y \Rightarrow G \, xy \matchpeg y$ that we proved true for LL(1) PE-CFGs without $\varepsilon$ expressions is false for grammars with $\varepsilon$ expressions, as the following simple grammar shows:
\[
S \rightarrow \Choice{a}{\varepsilon}
\]

This grammar is LL(1), and we have $G \, a \matchg a$ through rule {\bf choice.2}, but $G \, a \not\matchpeg a$, although simple inspection shows that the language of $G$ is $\{ a, \varepsilon \}$ whether interpreted as a CFG or as a PEG. 

We solve the above problem by introducing an end-of-input marker \$ ($\$ \notin T$), and using this marker to constrain proof trees so we only consider trees that consume the input and leave just the marker. Instead of trying to prove that $G \, xy \matchg y \Rightarrow G \, xy \matchpeg y$, we will prove that $G \, x\$ \matchg \$ \Rightarrow G \, x\$ \matchpeg \$ $, which will still be enough to prove that $G$ has the same language either interpreted as a PE-CFG or as a PEG.

A PE-CFG $G$ with $\varepsilon$ expressions is LL(1) if and only if the following two restrictions hold for every production $A \rightarrow p$ of $G$ and every choice $\Choice{p_1}{p_2}$ of $p$:
\begin{enumerate}
\item $\mathit{FIRST}^G(p_1) \cap \mathit{FIRST}^G(p_2) = \emptyset$
\item $\mathit{FIRST}^G(p_1) \cap \mathit{FOLLOW}^G(A) = \emptyset$ if $\varepsilon \in \mathit{FIRST}^G(p_2)$
\end{enumerate}

This is a direct restatement of the LL(1) restrictions for traditional CFGs~\cite{knuth,gtdpl}, and it is straightforward to show that a CFG $G$ is LL(1) if and only if its corresponding PE-CFG $G^\prime$ is LL(1).

We can now show that an LL(1) PE-CFG $G$ has the same language whether interpreted as a CFG or as a PEG. The proof is a corollary of the following lemma:

\begin{lemma}
\label{lemma:ll1peg1}
Given an LL(1) PE-CFG $G$, if there is a proof tree for $G \, x\$ \matchg \$ $ then, for every subtree $G[p] \, x^\prime\$ \matchg x^{\prime\prime}\$ $, we have that $G[p] \, x^\prime \$ \matchpeg x^{\prime\prime}\$ $.
\end{lemma}

\begin{proof} By induction on the height of the proof tree for $G[p] \, x^\prime \$ \matchg x^{\prime\prime}\$ $. The interesting case is {\bf choice.2}. For this case, we have $G[p_2] \, x^\prime \$ \matchg x^{\prime\prime}\$ $. Because of BNF structure, this is a subtree of $G[A] \, x^\prime \$ \matchg x^{\prime\prime}\$ $ for some non-terminal $A$. By the definition of {\em FOLLOW}, the first symbol $a$ of $x^{\prime\prime}\$ $ is in $\mathit{FOLLOW}^{G}(A)$. We now have two subcases, one where $x^\prime = x^{\prime\prime}$ and another where $x^\prime = bwx^{\prime\prime}$.

In the first subcase, we have $a \notin \mathit{FIRST}^G(p_1)$ by the second LL(1) restriction. So $G[p_1] \, x^{\prime\prime}\$ \not\matchg y$ and, by denial of the consequent of Lemma~\ref{cfgd:subset} and completeness of LL(1) grammars, $G[p_1] \, x^{\prime\prime}\$ \matchpeg \Nothing$. With the induction hypothesis and the application of {\bf ord.2} we have $G[\Choice{p_1}{p_2}] \, x^{\prime\prime}\$ \matchpeg x^{\prime\prime}\$$.

In the second subcase, where $x^\prime = bwx^{\prime\prime}$, we have $b \in \mathit{FIRST}^G(p_2)$, so $b \notin \mathit{FIRST}^G(p_1)$ by the first LL(1) restriction. The rest of the proof is similar to the first subcase.
\end{proof}

The proof that $G \, x\$ \matchg \$ $ if and only if $G \, x\$ \matchpeg \$ $ for any LL(1) PE-CFG $G$ is now trivial, from the above lemma and from Lemma~\ref{cfgd:subset}.

In the next section we will show how any strong-LL($k$) grammar can be translated to a PEG that recognizes the same language, while keeping the overall structure of the grammar.

\section{Strong-LL($k$) Grammars and PEGs}
\label{sec:llkpeg}

Strong-LL($k$) grammars are a subset of CFGs where a top-down parser can predict which production to use for a non-terminal just by examining the next $k$ symbols of the input, where $k$ is arbitrary but fixed for each grammar. They are a special case of LL($k$) grammars, in which the parser can use both the next $k$ symbols of the input and the history of which productions it already picked during parsing.

Unlike LL(1) grammars, there are strong-LL($k$) grammars that have different languages when interpreted as CFGs and as PEGs, no matter how we order their choice expressions. For example, take the PE-CFG $G$ with the following productions:
\[
S \rightarrow \Choice{A}{B} \tenspaces
A \rightarrow \Choice{ab}{C} \tenspaces
B \rightarrow \Choice{a}{Cd} \tenspaces
C \rightarrow c
\]

$G$ is a strong-LL(2) grammar, and its language, when interpreted as a CFG, is $\{ a, ab, c, cd \}$. But interpreting $G$ as a PEG yields the language $\{ a, ab, c \}$; when matching $cd$, non-terminal $A$ succeeds (through its second alternative, non-terminal $C$), and non-terminal $B$ (the second alternative of $S$) is never tried. Changing $S$ to $S \rightarrow \Choice{B}{A}$ changes the PEG's language to $\{ a, c, cd \}$, which is still different from the language of $G$ as a CFG, because what happened to $cd$ now happens to $ab$.

Nevertheless, any strong-LL($k$) language can be parsed by a top-down parser without backtracking while using $k$ symbols of lookahead. So it seems intuitive that we can use syntactic predicates to direct a PEG parser to the right alternative. We cannot interpret $G$ as a PEG and recognize the same language, but we can add predicates to $G$, to emulate the predictions that a strong-LL($k$) parser makes.

An approach for translating a PE-CFG $G$ to a PEG that recognizes the same language is to add an and-predicate (syntactical sugar for a double application of the not-predicate) in front of every alternative of a non-terminal; this and-predicate tests the next $k$ symbols of the input against the possible lookahead values that a strong-LL($k$) parser would use for that alternative. For the strong-LL(2) grammar above, the translation results in the following PEG:
\begin{eqnarray*}
S & \rightarrow & \&(\Choice{ab}{c\$})\,A\ \ \ |\ \ \ \&(\Choice{a\$}{cd})\,B \tenspaces \fivespaces \\
A & \rightarrow & \&(ab)\,ab\ \ \ |\ \ \ \&(c\$)\,C \tenspaces \fivespaces \\
B & \rightarrow & \&(a\$)\,a\ \ \ |\ \ \ \&(cd)\,Cd \tenspaces \fivespaces \\
C & \rightarrow & \&(\Choice{cd}{c\$})\,c
\end{eqnarray*}

It is easy to check that this PEG recognizes the language $\{ a, ab, c, cd \}$, the same as $G$, if we include the marker \$ at the end of the input strings for the PEG. The formal definition of the translation does not add the predicate to the last alternative of a non-terminal (or to the sole alternative, in case of non-terminal $C$ above).

Before formalizing our translation and proving its correctness, we will give definitions of strong-LL($k$) properties using our new CFG formalism. First we need an auxiliary function $\mathit{take}_k$, with the definition below:
\begin{eqnarray*}
\mathit{take}_k(\varepsilon) & = & \varepsilon \\
\mathit{take}_k(a_1 \ldots a_n) & = & \left \{ \begin{array}{ll}
  a_1 \ldots a_k & \mbox{if $n > k$} \\
  a_1 \ldots a_n & \mbox{otherwise}
  \end{array}
  \right .
\end{eqnarray*}

We will say that $take_k(x)$ is the {\em $\mathrm{k}$-prefix} of $x$. We also need to define $\bullet_k$, a language concatenation operation that results in $k$-prefixes (i.e. concatenates each string of the first language with each string of the second language, taking the $k$-prefix of each result):
\[
X \bullet_k Y = \{ \mathit{take}_k(x) \,\, | \,\,  x \in X \cdot Y \}
\]

A property of $k$-prefixes is that the $k$-prefix of the concatenation of two strings is also the $k$-prefix of the concatenation of their $k$-prefixes (proof by case analysis on the definition of $\mathit{take}_k$):
\[
\mathit{take}_k(xy) = \mathit{take}_k(\mathit{take}_k(x)\mathit{take}_k(y))
\]

This leads directly to the following simple lemma, which we only include to reference in later proofs:

\begin{lemma}
\label{lemma:take}
If $\mathit{take}_k(x) \in X$ and $\mathit{take_k}(y) \in Y$ then we have $\mathit{take}_k(xy) \in X \bullet_k Y$.
\end{lemma}

\begin{proof}
Trivial.
\end{proof}

We can now define the $\mathit{FIRST}_k$ sets, the strong-LL($k$) analog of the LL(1) {\em FIRST} sets. The $\mathit{FIRST}_k$ set of an expression $p$ is the set of the $k$-prefixes of every string that $p$ matches: 
\[
\mathit{FIRST}^G_k(p) = \{ \mathit{take}_k(x) \,\, | \,\, G[p] \, xy \matchg y \}
\]

The definition of $\mathit{FOLLOW}_k$ sets is also a straightforward extension of the definition of {\em FOLLOW} sets for LL(1) grammars:
\[
\begin{array}{l}
\mathit{FOLLOW}^{G}_k(A) = \{ \mathit{take}_k(y) \, | \, \mbox{$G[A] \, xy \matchg y$ is in a }\\
\tenspaces \tenspaces \tenspaces \tenspaces \fivespaces \,\,\mbox{proof tree for $G \, w\$^k \matchg \$^k $} \}
\end{array}
\]

To ensure that all members of $\mathit{FOLLOW}_k$ have length $k$, we use $k$ end-of-input markers $\$ \notin T$ instead of the single marker we used with LL(1) grammars. The semantics of $\matchg$ guarantee that $\$^k$ is a suffix of $y$, so the length of $y$ is at least $k$ and the length of $\mathit{take}_k(y)$ is always $k$.

We can now state the strong-LL($k$) property: a PE-CFG $G$ with BNF structure is strong-LL($k$) if and only if every choice expression $\Choice{p_1}{p_2}$ of every production $A \rightarrow p$ satisfies the following condition:
\[
\begin{array}{l}
(\mathit{FIRST}^G_k(p_1) \bullet_k \mathit{FOLLOW}^G_k(A)) \,\cap \\
 (\mathit{FIRST}^G_k(p_2) \bullet_k \mathit{FOLLOW}^G_k(A)) = \emptyset
\end{array}
\]

The strong-LL($k$) property is just a formal way of saying that the next $k$ symbols of the input are enough to choose among the choice expressions of a given non-terminal.

We also need an auxiliary function {\em choice} that takes a set of strings and makes a choice expression with each string as an alternative of this choice:
\begin{eqnarray*}
\mathit{choice}(\emptyset) = \varepsilon \\
\mathit{choice}(\{ p_1, \ldots, p_n \}) & = & \Choice{p_1}{\Choice{\ldots}{p_n}} \end{eqnarray*}

We will use {\em choice} to transform a lookahead set into a lookahead expression. Our translation inserts lookahead expressions to direct the PEG parser to the correct alternative in a choice, so we only need to changes choice operations. Because we are assuming that our PE-CFGs have BNF structure, these choice operations are at the ``top-level'' of each production. Intuitively, if $\Choice{p_1}{p_2}$ is a choice of non-terminal $A$, $\varphi^G_k(\Choice{p_1}{p_2}, A)$ adds the lookahead expression $\mathcal{L}^G(p_1, A)$ to $p_1$ and recursively transforms $p_2$; any expression that is not a choice is not transformed:
\begin{eqnarray*}
\varphi^G_k(\Choice{p_1}{p_2}, A) & = & \Choice{\mathcal{L}^G(p_1, A)p_1\,}{\,\varphi^G_k(p_2, A)} \\
\varphi^G_k(\varepsilon, A) & = & \varepsilon \\
\varphi^G_k(a, A) & = & a \\
\varphi^G_k(p_1p_2, A) & = & p_1p_2 \\
\varphi^G_k(B, A) & = & B \\
\mbox{where} \ \mathcal{L}^G(p, A) & = & \&\mathit{choice}(\mathit{FIRST}^G_k(p) \bullet_k \\
& & \tenspaces \fivespaces \mathit{FOLLOW}^G_k(A))
\end{eqnarray*}

The definition of our translation now is straightforward. From a strong-LL($k$) grammar $G$ with BNF structure we can generate a PEG $\Phi_b(G)$ (the {\em before} LL($k$)-PEG of $G$) by replacing each production $A \rightarrow p$ with $A \rightarrow \varphi^G_k(p, A)$.

To prove the correctness of the translation, we will use the same approach that we took in the proof for LL(1) grammars with $\varepsilon$ expressions. We will prove that in any derivation of $G \, x\$^k \matchg \$^k$ all of the subparts of the derivation have correspondents in $\Phi_b(G)$ via function $\varphi^G_k$. One subtlety of the proof is the parameter $A$ of $\varphi^G_k$; our definition of $\Phi_b(G)$ makes it clear that $A$ in $\varphi^G_k(p, A)$ is the non-terminal that ``owns'' the expression $p$. For an expression $p$ that appears in a subpart of the derivation of $G \, x\$^k \matchg \$^k$ as $G[p]$, $A$ is the first non-terminal that appears as $G[A]$ in a path from this subpart to the conclusion $G \, x\$^k \matchg \$^k$. Formally, we can state the following lemma, a version of Lemma~\ref{lemma:ll1peg1}:

\begin{lemma}
\label{lemma:llkpeg1}
Given a strong-LL($k$) PE-CFG $G$, if there is a proof tree for $G \, x\$^k \matchg \$^k $ then, for every subtree $G[p] \, x^\prime\$^k \matchg x^{\prime\prime}\$^k $ of this proof tree, we have $\Phi_b(G)[\varphi^G_k(p, A)] \, x^\prime \$^k \matchpeg x^{\prime\prime}\$^k $, where $A$ is the first non-terminal that appears as $G[A]$ in a path from the conclusion $G[p] \, x^\prime\$^k \matchg x^{\prime\prime}\$^k $ of the subtree to the conclusion $G \, x\$^k \matchg \$^k $ of the whole tree.
\end{lemma}

\begin{proof} By induction on the height of the proof tree for $G[p] \, x^\prime \$^k \matchg x^{\prime\prime}\$^k $. The interesting cases are {\bf choice.1} and {\bf choice.2}. For case {\bf choice.1}, we have $G[p_1] \, x^\prime \$^k \matchg x^{\prime\prime}\$^k $. Because of BNF structure, this is a subtree of $G[A] \, x^\prime \$^k \matchg x^{\prime\prime}\$^k $, and we have $\mathit{take}_k(x^{\prime\prime}\$^k) \in \mathit{FOLLOW}^G_k(A)$ by the definition of $\mathit{FOLLOW}_k$. If we combine this with the Lemma~\ref{lemma:take} we have $\mathit{take}_k(x^\prime\$^k) \in \mathit{FIRST}^G_k(p_1) \bullet_k \mathit{FOLLOW}^G_k(A)$, because the $k$-prefix of what $p_1$ matches is in $\mathit{FIRST}^G_k(p_1)$. So  $\Phi_b(G)[\mathcal{L}^G(p_1, A)] \, x^\prime\$^k \matchpeg x^\prime\$^k $ by the definition of $\mathcal{L}^G$. By the induction hypothesis, $\Phi_b(G)[p_1] \, x^\prime\$^k \matchpeg x^{\prime\prime}\$^k$, and with applications of rules {\bf con.1} and {\bf ord.1} we have $\Phi_b(G)[\varphi^G_k(\Choice{p_1}{p_2}, A)] \, x^\prime \$^k \matchpeg x^{\prime\prime}\$^k $.

For case {\bf choice.2}, we can use the LL($k$) property and an argument similar to the one used in {\bf choice.1} to conclude that $\mathit{take}_k(x^\prime\$^k) \notin \mathit{FIRST}^G_k(p_1) \bullet_k \mathit{FOLLOW}^G_k(A)$. So, by the definition of $\mathcal{L}^G$, $\Phi_b(G)[\mathcal{L}^G(p_1, A)] \, x^\prime\$^k \matchpeg \Nothing$. By the induction hypothesis, we have $\Phi_b(G)[\varphi^G_k(p_2, A)] \, x^\prime\$^k \matchpeg x^{\prime\prime}\$^k$, and by rules {\bf con.2} and {\bf ord.2} we have $\Phi_b(G)[\varphi^G_k(\Choice{p_1}{p_2}, A)] \, x^\prime \$^k \matchpeg x^{\prime\prime}\$^k $.
\end{proof}

We also need to prove that for any strong-LL($k$) PE-CFG $G$ we have $G \, x\$^k \matchg \$^k $ if $\Phi_b(G) \, x\$^k \matchpeg \$^k $. Intuitively, if $(\&p_1)p_2$ matches a string $x$ then $p_2$ also matches $x$, so if we have a proof tree for $\Phi_b(G) \, x\$^k \matchpeg \$^k $ we will be able to erase all the predicates introduced by $\Phi_b$ and build a proof tree for $G \, x\$^k \matchg \$^k $. 

First, let us define predicate erasure as follows: the erasure of $(!p_1)p_2$ is the erasure of $p_2$. Any predicate occurring alone is replaced by $\varepsilon$. All other expressions just recursively erase predicates on their subparts. We get the erasure of a grammar by erasing the predicates in the right sides of every production, plus the initial symbol. The purpose of having a special case for the erasure of $(!p_1)p_2$ is to have the erasure of $\Phi_b(G)$ be $G$. Now we can prove the following lemma, which states that removing the predicates of a PEG
$G$ gives us a PE-CFG with a language that is a superset of the language of $G$:

\begin{lemma}
\label{llk:subsetb}
Given a PEG $G$ and an expression $p$, and the PE-CFG $G^\prime$ and expression $p^\prime$ obtained by erasing all predicates of $G$ and $p$, if $G[p] \, xy \matchpeg y$ then $G^\prime[p^\prime] \, xy \matchg y$.
\end{lemma}

\begin{proof} By induction on the height of the proof tree for $G[p] \, xy \matchpeg y$.
\end{proof}

The proof that $\Phi_b(G)$ has the same language as $G$ is now a corollary of Lemmas~\ref{lemma:llkpeg1} and~\ref{llk:subsetb}. This lemma will also be useful in the rest of this section and in the following one, to prove the correctness of our other transformations.

There is another approach for translating a strong-LL($k$) PE-CFG $G$ to an equivalent PEG. This approach uses a subtle consequence of the strong-LL($k$) property: take the alternatives $p_1$ to $p_n$ of a non-terminal $A$. Now let's say that two alternatives $p_i$ and $p_j$ both match prefixes of an input $w$, say $x_i$ and $x_j$, with $x_iy_i = x_jy_j = w$; that is, $G[p_i] \, x_iy_i \matchg y_i$ and $G[p_j] \, x_jy_j \matchg y_j$. By the definition of $\mathit{FIRST}_k$, we have $\mathrm{take}_k(x_i) \in \mathit{FIRST}^G_k(p_i)$ and $\mathrm{take}_k(x_j) \in \mathit{FIRST}^G_k(p_j)$. Therefore we cannot have both $\mathrm{take}_k(y_i) \in \mathit{FOLLOW}^G_k(A)$ and $\mathrm{take}_k(y_j) \in \mathit{FOLLOW}^G_k(A)$, or we would violate the strong-LL($k$) property by having $\mathrm{take}_k(w)$ in both $\mathit{FIRST}^G_k(p_i) \bullet_k \mathit{FOLLOW}^G_k(A)$ and $\mathit{FIRST}^G_k(p_j) \bullet_k \mathit{FOLLOW}^G_k(A)$ (Lemma~\ref{lemma:take}).

The fact that we cannot have the first $k$ symbols of both $y_i$ and $y_j$ in $\mathit{FOLLOW}^G_k(A)$ is the core of this other approach, which is to add a guard {\em after} each alternative of a non-terminal $A$ to test if the next $k$ symbols of the input are in $\mathit{FOLLOW}^G_k(A)$. PEG's local backtracking then guarantees that the wrong alternative will not be taken even if it matches a prefix of the input.

For the strong-LL(2) grammar we used as an example in the beginning of this section, this approach yields the following translated PEG:
\begin{eqnarray*}
S & \rightarrow & A\,\&(\$\$)\ \ |\ \ B\,\&(\$\$) \tenspaces \fivespaces \\
A & \rightarrow & ab\,\&(\$\$)\ \ |\ \ C\,\&(\$\$) \tenspaces \fivespaces \\
B & \rightarrow & a\,\&(\$\$)\ \ |\ \ Cd\,\&(\$\$) \tenspaces \fivespaces \\
C & \rightarrow & c\,\&(\Choice{\$\$}{d\$})
\end{eqnarray*}

It is easy to check that this PEG recognizes the correct language $\{ a, ab, c, cd \}$ if we include the marker \$\$ at the end of the input string.

Like with our first translation, our second translation uses a function $\phi^G_k$ that translates the choice expressions in the production of a non-terminal $A$, adding an and-predicate built from a choice of every string in $\mathit{FOLLOW}^G_k(A)$ to the first half of the choice and recursively translating the second half. As with $\varphi^G_k$, $\phi^G_k$ is the identity function for other kinds of expressions, as the translation only changes choice expressions and we assume BNF structure:
\begin{eqnarray*}
\phi^G_k(\Choice{p_1}{p_2}, A) & = & p_1\&\mathit{choice}(\mathit{FOLLOW}^G_k(A)) \ \  | \\
& & \fivespaces \phi^G_k(p_2, A) \\
\phi^G_k(\varepsilon, A) & = & \varepsilon \\
\phi^G_k(a, A) & = & a \\
\phi^G_k(p_1p_2, A) & = & p_1p_2 \\
\phi^G_k(B, A) & = & B
\end{eqnarray*}

The definition of the second translation is now straightforward. From a strong-LL($k$) grammar $G$ with BNF structure we can generate a PEG $\Phi_a(G)$ (the {\em after} LL($k$)-PEG of $G$) by replacing each production $A \rightarrow p$ with $A \rightarrow \phi^G_k(p, A)$. 

The following lemma is like Lemma~\ref{lemma:llkpeg1} in that it proves that all subparts of a derivation for $G \, x\$^k \matchg \$^k$ have correspondents in $\Phi_a(G)$ via function $\phi^G_k$. As with lemmas~\ref{lemma:llkpeg1} and~\ref{lemma:ll1peg1}, we need to restrict ourselves to matches that consume all input but the end-of-input marker $\$^k$ so we can use the $\mathit{FOLLOW}_k$ sets of the non-terminals in our proof, and by extension the LL($k$) properties of $G$.

\begin{lemma}
\label{lemma:llkpeg2}
Given a strong-LL($k$) PE-CFG $G$, if there is a proof tree for $G \, x\$^k \matchg \$^k $ then, for every subtree $G[p] \, x^\prime\$^k \matchg x^{\prime\prime}\$^k $ of this proof tree, we have $\Phi_a(G)[\phi^G_k(p, A)] \, x^\prime \$^k \matchpeg x^{\prime\prime}\$^k $, where $A$ is the first non-terminal that appears as $G[A]$ in a path from the conclusion $G[p] \, x^\prime\$^k \matchg x^{\prime\prime}\$^k $ of the subtree to the conclusion $G \, x\$^k \matchg \$^k $ of the whole tree.
\end{lemma}

\begin{proof} By induction on the height of the proof tree for $G[p] \, x^\prime \$^k \matchg x^{\prime\prime}\$^k $. The interesting cases are {\bf choice.1} and {\bf choice.2}. For case {\bf choice.1}, we have $G[p_1] \, x^\prime \$^k \matchg x^{\prime\prime}\$^k $. Because of BNF structure, this is a subtree of $G[A] \, x^\prime \$^k \matchg x^{\prime\prime}\$^k $, and $\mathit{take}_k(x^{\prime\prime}\$^k) \in \mathit{FOLLOW}^G_k$ by the definition of $\mathit{FOLLOW}_k$. It is then easy to see that $\Phi_a(G)[\&\mathit{choice}(\mathit{FOLLOW}^G_k(A))] \, x^{\prime\prime}\$^k \matchpeg x^{\prime\prime}\$^k$. By the induction hypothesis, we have $\Phi_a(G)[p_1] \, x^\prime\$^k \matchpeg x^{\prime\prime}\$^k$, and with applications of {\bf con.1} and {\bf ord.1} we have $\Phi_a(G)[\phi^G_k(\Choice{p_1}{p_2}, A)] \, x^\prime \$^k \matchpeg x^{\prime\prime}\$^k $.

For case {\bf choice.2}, if there is no $w$ so $\Phi_a(G)[p_1] \, x^\prime\$^k \matchpeg w$ then we can conclude $\Phi_a(G)[p_1] \, x^\prime\$^k \matchpeg \Nothing$ by completeness of LL($k$) grammars, and we can apply the induction hypothesis on $p_2$ and rules {\bf con.2} and {\bf ord.2} to get $\Phi_a(G)[\phi^G_k(\Choice{p_1}{p_2}, A)] \, x^\prime \$^k \matchpeg x^{\prime\prime}\$^k $. Now suppose we have $\Phi_a(G)[p_1] \, x^\prime\$^k \matchpeg w$. Expression $p_1$ is predicate-free, so we have $G[p_1] \, x^\prime\$^k \matchg w$ by Lemma~\ref{llk:subsetb}. We have $\mathrm{take}_k(x^{\prime\prime}\$^k) \in \mathit{FOLLOW}^G_k(A)$ by the definition of $\mathit{FOLLOW}_k$, , and we have already seen that means $\mathrm{take}_k(w) \notin \mathit{FOLLOW}^G_k(A)$ or the LL($k$) property is violated. So we have $\Phi_a(G)[\&\mathit{choice}(\mathit{FOLLOW}^G_k(A))] \, w \matchpeg \Nothing$, and we can apply the induction hypothesis on $p_2$ and rules {\bf con.1} and {\bf ord.2} to conclude $\Phi_a(G)[\phi^G_k(\Choice{p_1}{p_2}, A)] \, x^\prime \$^k \matchpeg x^{\prime\prime}\$^k $. 
\end{proof}

The proof that $\Phi_a(G)$ has the same language as $G$ is a corollary of Lemmas~\ref{lemma:llkpeg2} and~\ref{llk:subsetb}.

LL(1) grammars are a special case of LL($k$) grammars where $k = 1$, and every LL(1) grammar is also a strong-LL(1) grammar (and vice-versa)~\cite{gtdpl,parsing}, so the two transformations we presented can also be used for LL(1) grammars, although the lookahead expressions become redundant.

\section{Right-linear and LL-regular Grammars}

A {\em right-linear CFG} is one where the right side of every production has at most one non-terminal, and this non-terminal can only appear as the last symbol of the production. Right-linear CFGs can only define regular languages, and any regular language $R$ has a right-linear CFG $G$; it is straightforward to encode any NFA as a right-linear CFG, and vice-versa~\cite{undecamb}.

For PE-CFGs, we will define a {\em right-linear PE-CFG} as a CFG where the right side of every production is a {\em right-linear parsing expression}. We define right-linear parsing expressions as the following predicate on parsing expressions: $\varepsilon$, $a$ and $A$ are right-linear; $p_1p_2$ is right-linear if and only if $p_1$ is a terminal and $p_2$ is right-linear; $p_1\,|\,p_2$ is right-linear if and only if $p_1$ and $p_2$ are right-linear. It is easy to see that any right-linear CFG has a corresponding right-linear PE-CFG, and vice-versa, as the transformations between CFGs and PE-CFGs we have on Section 2 preserve right-linearity.

Right-linear PE-CFGs, in the general case, do not recognize the same language when interpreted by $\matchg$ and $\matchpeg$. An obvious example is the grammar $S \rightarrow a\,|\,aa$, which is right-linear and has the language $\{ a, aa \}$ under $\matchg$ but $\{ a \}$ under $\matchpeg$. But the simplicity of right-linear grammars lets us prove an equivalence between right-linear CFGs and PEGs by adding a single restriction: the grammar's language must have the {\em prefix property}, that is, there are no distinct strings $x$ and $y$ in the language such that $x$ is a prefix of $y$.

We could prove this equivalence directly, but it is more interesting to combine a few more
general lemmas that individually deal with the relation of the prefix property and PEGs, and
of right-linear grammars and the prefix property. The first lemma gives an equivalence
between a class of PE-CFGs that have a stronger form of the prefix property and PEGs:

\begin{lemma}
\label{lemma:prefix}
Given a PE-CFG $G$ and a parsing expression $p$ where, for any choice $\Choice{q}{r}$
in $G[p]$, $L(G[\Choice{q}{r}])$ has the prefix property, if $G[p] \, xy \matchg y$ 
then $G[p] \, xy \matchpeg y$.
\end{lemma}

\begin{proof}
By induction on the height of the proof tree for $G[p] \, xy \matchg y$. The interesting case
if {\bf choice.2}. We have $p = p_1\,|\,p_2$, so $L(G[p]) = L(G[p_1]) \cup L(G[p_2])$, and
both $L(G[p_1])$ and $L(G[p_2])$ have the prefix property. This means that either
$G[p_1] \, xy \matchg y$, so we can use the induction hypothesis and rule
{\bf ord.1} to get $G[p] \, xy \matchpeg y$, or there is no suffix $z$ of $xy$ with $G[p_1]
\, xy \matchg z$, as that would violate the prefix property of $L(G[p_1])$. In this case, 
we have $G[p_1] \, xy \matchpeg \Nothing$ by modus tollens of Lemma~\ref{cfgd:subset}, 
and we can now use the induction hypothesis with $p_2$, and rule {\bf ord.2}, to get $G[p] \, xy
\matchpeg y$.
\end{proof}

The second lemma distributes the prefix property to the components of a right-linear
parsing expression, when the language of the expression has the prefix property.:

\begin{lemma}
\label{lemma:rlprefix1}
Given a right-linear PE-CFG $G$ and a right-linear parsing expression $p$, if $L(G[p])$
has the prefix property then, for any parsing expression $q$ in $p$, $L(G[q])$ has the prefix
property.
\end{lemma}

\begin{proof}
By structural induction on $p$. The interesting case is $p = \Con{p_1}{p_2}$. As $p$ is right-linear, $p_1$ is a terminal and $p_2$ is right-linear, so $L(G[p_1])$ trivially has the
prefix property. As $L(G[p]) = L(G[p_1]) \cdot L(G[p_2])$, $L(G[p_2])$ must have the
prefix property, and we can use the induction hypothesis to conclude that any parsing expression
$L(G[q])$ will also have the prefix property for any parsing expression $q$ in $p_2$.
\end{proof}

The third lemma distributes the prefix property to all of the components of a
right-linear grammar, when the language of the starting expression has the prefix property:

\begin{lemma}
\label{lemma:rlprefix2}
Given a right-linear PE-CFG $G$ and a right-linear parsing expression $p$, if $L(G[p])$
has the prefix property then, for any parsing expression $q$ in $G[p]$, $L(G[q])$ has the prefix
property.
\end{lemma}

\begin{proof}
By induction on the number of steps necessary to reach $q$ from $p$: zero steps if $p = q$,
one step if $q$ is a part of $p$, $k+1$ steps if $q$ is a part of $P(A)$ where $A$ is
reachable from $p$ in $k$ steps.
\end{proof}

The final lemma uses the first and the third lemma, plus Lemma~\ref{cfgd:subset}, to
prove the equivalence between right-linear CFGs with the prefix property and PEGs:

\begin{lemma}
\label{lemma:rightlin}
Given a right-linear PE-CFG $G$ and a right-linear parsing expression $p$, if $L(G[p])$ has the prefix property then $G[p] \, xy \matchg y$ if and only if $G[p] \, xy \matchpeg y$.
\end{lemma}

\begin{proof}
By Lemma~\ref{lemma:rlprefix2}, the choice expression in $G[p]$ also has the prefix property,
so, by Lemma~\ref{lemma:prefix}, $G[p] \, xy \matchg y$ implies $G[p] \, xy \matchpeg y$.
By Lemma~\ref{cfgd:subset}, $G[p] \, xy \matchpeg y$ implies $G[p] \, xy \matchg y$.
\end{proof}

The prefix-property restriction may seem overly restrictive, but we can obtain a right-linear grammar with the prefix property from any right-linear grammar $G$ by applying the following transformation, where we use $\$ \notin T$ as an end-of-input marker, to the right side of $G$'s productions and to its initial parsing expression:

\begin{eqnarray*}
\Pi(\varepsilon) & = & \$ \\
\Pi(a) & = & a\$ \\
\Pi(A) & = & A \\
\Pi(p_1p_2) & = & p_1\Pi(p_2) \\
\Pi(p_1 \, | \, p_2) & = & \Pi(p_1) \, | \, \Pi(p_2)
\end{eqnarray*}

A (strong) LL-regular CFG~\cite{nijholt,jazarbek} is a generalization of LL($k$) CFGs where a predictive top-down parser may decide which alternative to take based on where the rest of the input falls on a set of regular partitions of $T^*$. Formally, a CFG $G$ is LL-regular if there is a regular partition $\pi$ of $T^*$ such that for any two leftmost derivations of the following form,  if $x \equiv y \, (\mathrm{mod}\, \pi)$ then $\gamma = \delta$:

\begin{eqnarray*}
& S \stackrel{*}{\Rightarrow}_G w_1A\alpha_1 \Rightarrow_G w_1\gamma\alpha_1 \stackrel{*}{\Rightarrow}_G w_1x & \\
& S \stackrel{*}{\Rightarrow}_G w_2A\alpha_2 \Rightarrow_G w_1\delta\alpha_2 \stackrel{*}{\Rightarrow}_G w_2y &
\end{eqnarray*}

Restating the definition for PE-CFGs is straightforward. First we introduce the $\mathit{BLOCK}^G_\pi$ set that tells in which blocks of partition $\pi$ the input for $p$ falls:
\[
\begin{array}{l}
\mathit{BLOCK}^G_\pi(p, A) =  \{ B_k \in \pi \, | \, \mbox{$G[p] \, xy \matchg y$ and $G[A] \, xy \matchg y$}\\
\tenspaces \tenspaces \tenspaces \tenspaces \,\,\mbox{are in a proof tree for $G \, w\$ \matchg \$ $} \\
\tenspaces \tenspaces \tenspaces \tenspaces \,\,\mbox{and $xy \in B_k$}\}
\end{array}
\]

A PE-CFG $G$ with BNF structure is LL-regular if and only if there is a partition $\pi$ of $T^* \cdot \{ \$ \}$ such that every choice expression $p_1 \, | \, p_2$ of every production $A \rightarrow p$ has $\mathit{BLOCK}^G_\pi(p_1, A) \cap \mathit{BLOCK}^G_\pi(p_2, A) = \emptyset$.

The original definition of LL-regular grammars uses the partition where the rest of the input falls to predict alternatives, so our addition of an end-of-input marker \$ is not changing the class of grammars we are defining, while ensuring that the blocks of the regular partition have the prefix property.

Any strong-LL($k$) grammar is also LL-regular~\cite{nijholt82}, so a simple reordering of alternatives is not sufficient for obtaining a PEG that recognizes the same language as an LL-regular grammar $G$. But we can use the same approach we used in the translation $\Phi_b$ to translate an LL-regular grammar $G$ into a PEG \RE{G} that recognizes the same language, using an and-predicate to add a lookahead expression to the front of the alternatives of each choice.

We assume that each block $B_k$ of the regular partition $\pi$ has a corresponding right-linear grammar $G_{B_k}$, where the intersection of the non-terminal sets of $G$ and of all these grammars is empty; the non-terminal set of \RE{G} is the union of these sets.

We form the regular lookahead of an alternative $p$, $\mathcal{L}^G_r(p, A)$, by making a choice of the grammars for each block in $\mathit{BLOCK}^G_\pi(p, A)$, and wrapping this choice in an and-predicate:
\[
\mathcal{L}^G_r(p, A) = \&\mathit{choice}(\{S_{B_k} | B_k \in \mathit{BLOCK}^G_\pi(p, A)\})
\]
where {\em choice} is the function we used to build the lookahead expressions for $\varphi_b$ and $\phi_b$.

Function $\rho^G(p, A)$ adds lookahead expressions where necessary, assuming that the original grammar has BNF structure:
\begin{eqnarray*}
\rho^G(\Choice{p_1}{p_2}, A) & = & \Choice{\mathcal{L}^G_r(p_1, A)p_1\,}{\,\rho^G(p_2, A)} \\
\rho^G(\varepsilon, A) & = & \varepsilon \\
\rho^G(a, A) & = & a \\
\rho^G(p_1p_2, A) & = & p_1p_2 \\
\rho^G(B, A) & = & B
\end{eqnarray*}

We obtain the productions of \RE{G} by applying $\rho^G$ to the right-side of each production of $G$, and then adding the productions for each $G_{B_k}$. We can prove a lemma similar to Lemma~\ref{lemma:llkpeg1}:

\begin{lemma}
\label{lemma:llreg}
Given an LL-regular PE-CFG $G$, if there is a proof tree for $G \, x\$ \matchg \$ $ then, for every subtree $G[p] \, x^\prime\$ \matchg x^{\prime\prime}\$ $, we have $\RE{G}[\rho^G(p, A)] \, x^\prime \$ \matchpeg x^{\prime\prime}\$ $, where $A$ is the first non-terminal that appears as $G[A]$ in a path from the conclusion $G[p] \, x^\prime\$ \matchg x^{\prime\prime}\$ $ of the subtree to the conclusion $G \, x\$ \matchg \$ $ of the whole tree.
\end{lemma}

\begin{proof} By induction on the height of the proof tree for $G[p] \, x^\prime \$ \matchg x^{\prime\prime}\$ $. The interesting cases are {\bf choice.1} and {\bf choice.2}. For {\bf choice.1}, we have $G[p_1] \, x^\prime \$ \matchg x^{\prime\prime}\$ $. Because of the BNF structure, this is a subtree of $G[A] \, x^\prime \$ \matchg x^{\prime\prime}\$ $, and we have $x^\prime \$ $ in some block $B_k \in \mathit{BLOCK}^G_\pi(p_1, A)$. So $x^\prime \$ \in L(G_{B_k})$. It is easy to see that $x^\prime \$ \in L(\mathit{choice}(\{S_{B_k} | B_k \in \mathit{BLOCK}^G_\pi(p, A)\}))$; this grammar is right-linear, so by Lemma~\ref{lemma:rightlin} and the semantics of the and-predicate we have $\RE{G}[\mathcal{L}^G_r(p_1, A)] \, x^\prime \$ \matchpeg x^\prime \$ $. We have $\RE{G}[p_1] \, x^\prime\$ \matchpeg x^{\prime\prime}\$ $ by the induction hypothesis, and can use rules {\bf con.1} and {\bf ord.1} to get $\RE{G}[\rho^G(\Choice{p_1}{p_2}, A)] \, x^\prime \$ \matchpeg x^{\prime\prime}\$ $.

  For case {\bf choice.2}, we can use the LL-regular property to conclude that the block $B_k$ with $x^\prime \$ \in B_k$ is not in $\mathit{BLOCK}^G_\pi(p_1, A)$, and then use the definition of $\mathcal{L}^G_r$ and an argument similar to the one used in {\bf choice.1} to get $\RE{G}[\mathcal{L}^G_r(p_1, A)] \, x^\prime\$ \matchpeg \Nothing$. We can use the induction hypothesis to get $\RE{G}[\rho^G(p_2, A)] \, x^\prime\$ \matchpeg x^{\prime\prime}\$$, and then use rules {\bf con.2} and {\bf ord.2} to get $\RE{G}[\rho^G(\Choice{p_1}{p_2}, A)] \, x^\prime \$ \matchpeg x^{\prime\prime}\$ $.
\end{proof}

The proof that $\RE{G}$ has the same language as $G$ is a corollary of Lemmas~\ref{lemma:llreg} and~\ref{llk:subsetb}. We can use Lemma~\ref{llk:subsetb} even though $\RE{G}$ has non-terminals that are not present in $G$; these extra non-terminals are only referenced inside predicates, so they become useless when the predicates are removed, and can also be removed.

\section{Related Work}

Parsing Expression Grammars have generated much academic interest since their introduction by Ford~\cite{pegford}, with over sixty citations of Ford's paper in ACM's Digital Library. But just a few of these works are concerned with the theory of PEGs and their relation to other parsing tools; this section discusses these works and how they relate to our work.

Ford~\cite{pegford} leaves open the problem of how PEGs and CFGs relate, and does not outline a strategy to solve this problem. The solution of this problem for the major classes of top-down CFGs is a contribution of our work, and our recasting of the CFGs in a recognition-based formalism is another contribution that shows where CFGs and PEGs diverge.

This paper is an extension of previous, unpublished work done~\cite{sergio:tese}. This
older work has the first version of the PE-CFG semantics, as well a the transformations
between strong-LL($k$) grammars and PEGs, and the proof of the correspondence
between LL(1) grammars and PEGs. The proofs have been revised and improved.

Redziejowski~\cite{redz09} adapts the {\em FIRST} and {\em FOLLOW} relations of CFGs to the study of Parsing Expression Grammars, defining PEG analogs of these two relations. Redziejowski then uses the analogs for a conservative approximation of when an ordered choice is commutative or not, by defining a PEG analog of the LL(1) restriction. He admits that the approximation is too conservative for practical use, specially because of its treatment of syntactic predicates. We do not attempt to give definitions of {\em FIRST} and {\em FOLLOW} for PEGs, limiting our redefinitions of these relations just to our PE-CFG formalism, and using them to prove correspondences between LL(1) and strong-LL($k$) CFGs and structurally similar PEGs, something that Redziejowski does not explore.

An earlier work by Redziejowski~\cite{redz08} presents several identities regarding the languages defined by PEGs, although the author concludes that the identities are only useful for obtaining approximations of a PEGs language, and he also does not try to relate the languages of PEGs and of CFGs.

In a more recent work~\cite{redz13}, Redziejowski, based on our earlier, unpublished
work~\cite{sergio:tese}, extends our correspondence between LL(1) grammars and PEGs
to a larger class of grammars, and calls these grammars LL(1p). He notes that checking
whether a grammar is LL(1p) is harder than checking if the grammar is LL(1). The paper
also uses earlier versions of the new semantics for CFGs and PEGs that we presented
in Section 2.

Schmitz~\cite{sch06} presents both an ambiguity detection algorithm for CFGs in the context of the SDF2 formalism, an extension of CFGs, and an algorithm for detecting whether an ordered choice in a PEG is commutative or not. Schmitz notes that an overly strict ambiguity detector can be as restrictive as introducing ordering, but does not attempt to further study the relation between CFGs and PEGs.

Parr and Quong~\cite{parr94} add semantic and syntactic predicates to LL($k$) grammars to get the {\em pred-LL($k$)} parsing strategy. The syntactic predicates of pred-LL($k$) are only used in productions that have LL($k$) conflicts. In these productions, the parser tries to match the predicates of each alternative in the order they are given in the grammar definition, choosing the first alternative with a predicate that succeeds. Backtracking is strictly local, as with PEGs, so a subsequent failure does not make the parser try other productions. The paper does not give a formal specification of pred-LL($k$) grammars, nor how they relate to the class of LL($k$) grammars.

Parr and Fisher~\cite{parr11} introduce the {\em LL($*$)} parsing strategy, which uses the basic idea of LL-regular grammars, with a predictive top-down parser for these grammars that uses a deterministic finite automata to select which alternative of a non-terminal to take. If the grammar is not LL-regular, the LL($*$) parser can use semantic and syntactic predicates with local backtracking, as in pred-LL($k$), and also automatically introduce predicates, via a suitably named ``PEG mode''.

Generalized LL parsing~\cite{scott:glltree,scott:gll} extends the idea of LL(1) recursive
descent parsing to the full class of context-free grammars (including ambiguous grammars),
by making the parser proceed among the different conflicting alternatives ``in parallel''.
It replaces the call stack of a recursive descent parser with a {\em graph structured stack}
(GSS), a data structure adapted from Generalized LR parsers~\cite{tomita:gss}.
As the resulting parsers can parse ambiguous grammars, construction of the parse tree
is non-trivial to implement efficiently~\cite{scott:gllmodel,scott:glltree}, 
using a {\em shared packed parse forest} data structure also adapted from
generalized bottom-up techniques~\cite{scott:brnglr}.

Even if we restrict our domain to unambiguous grammars, we cannot use GLL parsing
as a basis for a correspondence between CFGs and PEGs; a GLL parser is not
a predictive parser, so we cannot try to encode the predictive part of the parser in a syntactic
predicate, as we did for LL-regular grammars, nor we can exploit a property of the
grammars, as we did for LL(1) and strong-LL($k$) grammars.

\section{Conclusions}

We presented a new formalism for context-free grammars that is based on recognizing (parts of) strings instead of generating them. We adopted a subset of the syntax of parsing expression grammars, and the notion of letting a grammar recognize just part of an input string, to purposefully get a definition for CFGs that is closer to PEGs, yet defines the same class of languages as traditional CFGs. These PE-CFGs define the same class of language as traditional CFGs, and simple transformations lets us get a PE-CFG from a CFG and vice-versa.

Our semantics for PE-CFGs has a non-deterministic choice operation. We showed how a deterministic choice operation based on the notions of failure and ordering turns PE-CFGs into quasi-PEGs; the addition of a {\em not} syntactic predicate then gave us a semantics that is equivalent to Ford's original semantics for PEGs. We then used our new formulations of CFGs and PEGs to study correspondences between four classes of CFGs and PEGs: LL(1), strong-LL($k$), right-linear and LL-regular. We proved that LL(1) grammars already define the same language either interpreted as CFGs or as PEGs, as was already suspected, and gave transformations that yield equivalent PEGs for the grammars in the other three classes. 

All our transformations preserve the structure of the original grammars; we do not change or remove non-terminals, nor change the alternatives of a non-terminal, just add predicates to the beginning or the end of each alternative. This means that our transformations have a practical application in reusing grammars made for one formalism with tools made for the other, as conserving the structure of the grammar makes it easier to carry semantic actions from one tool to another without modification.

Our transformations assume the presence of some kind of end-of-input marker, and incorporate this marker in the resulting PEG. This does not affect our equivalence results, but it does have implications in composability of PEGs resulting from our transformations. To use a PEG obtained from one of our transformations as part of a larger PEG (for embedding one language in another, for example) requires a suitable ``end-of-input marker'' to be picked (the boundaries between languages have to be explicit). We believe this problem should be easily solvable in practice.

A possible next step of this work would be the study of relationship between PEGs
and classic bottom-up CFGs, such as LR($k$)~\cite{lrk:knuth} and Simple
LR($k$)~\cite{slrk:deremer}. A key issue regarding this study is the fact that bottom-up
grammars can have left-recursive rules, but a PEG with left-recursive rules is
not complete~\cite{pegford}.

Recently, it was suggested a conservative extension of PEGs' semantics
that gives meaning for left-recursive rules, so PEGs with these rules
would also be complete~\cite{left:sblp}. Based on this PEGs' extesion we
could try to establish a correspondence between bottom-up grammars and
PEGs, and see if it is possible to achieve an equivalent PEG, with a
similar structure, from a bottom-up CFG.

\bibliographystyle{elsarticle-num}
\bibliography{paper-scp2}

\end{document}